\documentclass[11pt]{article}
\usepackage[dvipsnames,svgnames,x11names,hyperref]{xcolor}
\usepackage{amsmath}
\usepackage{amsthm}
\usepackage[T1]{fontenc}

\usepackage{fullpage}
\usepackage{thmtools}
\usepackage[colorlinks]{hyperref}
\hypersetup{
  linkcolor=[rgb]{0,0,0.4},
  citecolor=[rgb]{0, 0.4, 0},
  urlcolor=[rgb]{0.6, 0, 0}
}
\usepackage{amssymb,amsfonts,amsmath,amsthm,amscd,dsfont,mathrsfs,bbm}
\usepackage{todonotes}
\usepackage{complexity}
\usepackage{tikz}
\usetikzlibrary{decorations.pathreplacing,backgrounds}
\usepackage{multirow}

\usepackage{amsthm}
\usepackage{thmtools,thm-restate}

\numberwithin{equation}{section}
\declaretheoremstyle[bodyfont=\it,qed=\qedsymbol]{noproofstyle}

\declaretheorem[name=Observation,numbered=no]{observation*}

\declaretheorem[numberlike=equation]{theorem}

\declaretheorem[name=Theorem,numbered=no]{theorem*}

\declaretheorem[numberlike=equation]{lemma}
\declaretheorem[name=Lemma,numbered=no]{lemma*}

\declaretheorem[numberlike=equation]{corollary}
\declaretheorem[name=Corollary,numbered=no]{corollary*}

\declaretheorem[name=Proposition,numbered=no]{proposition*}

\declaretheorem[name=Claim,numbered=no]{claim*}

\declaretheorem[name=Conjecture,numbered=no]{conjecture*}

\declaretheorem[name=Question,numbered=no]{question*}

\declaretheoremstyle[bodyfont=\it,qed=$\lozenge$]{defstyle} 

\declaretheorem[numberlike=equation,style=defstyle]{definition}
\declaretheorem[unnumbered,name=Definition,style=defstyle]{definition*}

\declaretheorem[unnumbered,name=Example,style=defstyle]{example*}

\declaretheorem[unnumbered,name=Notation=defstyle]{notation*}

\declaretheorem[unnumbered,name=Construction,style=defstyle]{construction*}

\declaretheorem[numberlike=equation,style=defstyle]{remark}
\declaretheorem[unnumbered,name=Remark,style=defstyle]{remark*}

\usepackage{nth}
\usepackage{intcalc}
\usepackage{etoolbox}
\usepackage{xstring}

\usepackage{ifpdf}
\ifpdf
\else
\usepackage[quadpoints=false]{hypdvips}
\fi

\newcommand{\ECCCurl}[2]{http://eccc.hpi-web.de/report/\ifnumcomp{#1}{>}{93}{19}{20}#1/#2/}

\newcommand{\shortECCC}[2]{\texttt{\href{http://eccc.hpi-web.de/report/\ifnumcomp{#1}{>}{93}{19}{20}#1/#2/}{eccc:TR#1-#2}}}

\newcommand{\parseECCC}[1]{
\StrSubstitute{#1}{TR}{}[\tmpstring]%
\IfSubStr{\tmpstring}{/}{ 
\StrBefore{\tmpstring}{/}[\ecccyear]%
\StrBehind{\tmpstring}{/}[\ecccreport]%
}{
\StrBefore{\tmpstring}{-}[\ecccyear]%
\StrBehind{\tmpstring}{-}[\ecccreport]%
}%
\shortECCC{\ecccyear}{\ecccreport}}

\newcommand{\F}{\mathbb{F}}
\newcommand{\N}{\mathbb{N}}
\newcommand{\Q}{\mathbb{Q}}
\renewcommand{\C}{\mathbb{C}}

\newcommand{\va}{\mathbf{a}}
\newcommand{\vu}{\mathbf{u}}
\newcommand{\vv}{\mathbf{v}}
\newcommand{\vx}{\mathbf{x}}
\newcommand{\vy}{\mathbf{y}}
\newcommand{\vs}{\mathbf{s}}
\newcommand{\vw}{\mathbf{w}}

\newcommand{\valpha}{\boldsymbol{\alpha}}
\newcommand{\vbeta}{\boldsymbol{\beta}}

\newcommand{\vgamma}{\boldsymbol{\gamma}}
\newcommand{\SV}{\mathrm{SV}}

\newcommand{\rank}{\mathsf{rank}}

\newcommand{\Ideal}{\mathbf{I}}
\newcommand{\Variety}{\mathbf{V}}

\newcommand{\PIT}{\mathsf{PIT}}

\title{A Polynomial Degree Bound on Equations for Non-rigid Matrices and  Small Linear Circuits}
\author{Mrinal Kumar\thanks{Department of Computer Science \& Engineering, IIT Bombay. Email:\texttt{ mrinal@cse.iitb.ac.in}}
\and
 Ben Lee Volk \thanks{Center for the Mathematics of Information, California Institute of Technology, USA. Email:\texttt{ benleevolk@gmail.com} }}

\date{}

\begin{document}

\maketitle

\abstract{
We show that there is an equation of degree at most $\poly(n)$ for the (Zariski closure of the) set of the non-rigid matrices: that is, we show that for every large enough field $\F$, there is a non-zero $n^2$-variate polynomial $P \in \F[x_{1, 1}, \ldots, x_{n, n}]$ of degree at most $\poly(n)$ such that  every matrix $M$ which can be written as a sum of a matrix of rank at most $n/100$ and a matrix of sparsity at most $n^2/100$ satisfies $P(M) = 0$. This  confirms a conjecture of Gesmundo, Hauenstein, Ikenmeyer and Landsberg \cite{ GHIL16} and improves the best upper bound known for this problem down from $\exp(n^2)$ \cite{KLPS14, GHIL16} to $\poly(n)$. 

We also show a similar polynomial degree bound for the (Zariski closure of the) set of all matrices $M$ such that the linear transformation represented by $M$ can be computed by an algebraic circuit with at most  $n^2/200$ edges (without any restriction on the depth). As far as we are aware, no such bound was known prior to this work when the depth of the circuits is unbounded. 

Our methods are elementary and short and rely on a polynomial map of Shpilka and Volkovich \cite{SV15} to construct low degree ``universal'' maps for non-rigid matrices and small linear circuits. Combining this construction with a simple dimension counting argument to show that any such polynomial map has a low degree annihilating polynomial completes the proof. 

As a corollary, we show that any derandomization of the polynomial identity testing problem will imply new circuit lower bounds. A similar (but incomparable) theorem was proved by Kabanets and Impagliazzo \cite{KI04}.  
}

\section{Introduction}

\subsection{Equations for varities in algebraic complexity theory}
Let $V \subseteq \F^n$ be a (not necessarily irreducible) affine variety and let $\Ideal(V)$ denote its ideal.\footnote{For completeness, we provide the formal (standard) definitions for these notions in \autoref{sec:ag}.}. A non-zero polynomial $P \in \Ideal(V)$ is called an \emph{equation} for $V$. An equation for $V$ may serve as a ``proof'' that a point $\vx \in \F^n$ is \emph{not} in $V$, by showing that $P(\vx) \neq 0$.

A fundamental observation of the Geometric Complexity Theory program is that many important circuit lower bounds problems in algebraic complexity theory fit naturally into the setting of showing that a point $\vx$ lies outside a variety $V$ \cite{MulmuleyS01, BIPLS19}. In this formulation, one considers $V$ to be the closure of a class of polynomials of low complexity, and $\vx$ is the coefficient vector of the candidate hard polynomial.

Let $\Delta(V) := \min_{0 \neq P \in \Ideal(V)} \{ \deg(P) \}$. The quantity $\Delta(V)$ can be thought of as a measure of complexity for the geometry of the variety $V$. The quantity $\Delta(V)$ is a very coarse complexity measure. A recent line of work regarding \emph{algebraic natural proofs} \cite{FSV18, GKSS17} suggests to study the arithmetic circuit complexity of equations for varieties $V$ that correspond to polynomials with small circuit complexity. Having $\Delta(V)$ growing like a polynomial in $n$ is a necessary (but not a sufficient) condition for a variety $V$ to have an algebraic natural proof for non-containment.

\subsection{Rigid matrices}

A matrix $M$ is $(r,s)$-rigid if $M$ cannot be written as a sum $R+S$ where $\rank(R) \le r$ and $S$ contains at most $s$ non-zero entries. Valiant \cite{Valiant77} proved that if $A$ is $(\varepsilon n, n^{1+\delta})$-rigid for some constants $\varepsilon, \delta>0$ then $A$ cannot be computed by arithmetic circuits of size $O(n)$ and depth $O(\log n)$, and posed the problem of  \emph{explicitly} constructing rigid matrices with these parameters, which is still open. It is easy to prove that most matrices have much stronger rigidity parameters: over algebraically closed fields a generic matrix is $(r, (n-r)^2)$-rigid for any target rank $r$.

Let $\F$ be an algebraically closed field. Let $A_{r,s} \subseteq \F^{n \times n}$ denote the set of matrices which are not $(r, s)$-rigid. Let $V_{r,s} = \overline{A_{r,s}}$ denote the Zariski closure of $A_{r,s}$. A geometric study of $V_{r,s}$ was initiated by Kumar, Lokam, Patankar and Sarma \cite{KLPS14}. Among other results, they prove that for every $s < (n-r)^2$, $\Delta(V_{r,s}) \le n^{4n^2}$. A slightly improved (but still exponential) upper bound was obtained by Gesmundo, Hauenstein, Ikenmeyer and Landsberg \cite{GHIL16}, who also conjectured that for some $\varepsilon, \delta>0$, $\Delta(V_{\varepsilon n, n^{1+\delta}})$ grows like a polynomial function in $n$. The following theorem which we prove in this note confirms this conjecture. 
\begin{theorem}
\label{thm:poly-deg-bound}
Let $\varepsilon < 1/25$, and let $\F$ be a field of size at least $n^2$. For every large enough $n$, there exists a non-zero polynomial $Q \in \F[x_{1,1}, \ldots, x_{n,n}]$, of degree at most $n^3$, which is a non-trivial equation for matrices which are \emph{not} ($\varepsilon n$, $\varepsilon n^2$)-rigid. That is, for every such matrix $M$, $Q(M) = 0$.
\end{theorem}

 In fact, the conjecture of \cite{GHIL16} was slightly weaker: they conjectured that $\Delta(U)$ is polynomial in $n$ for every irreducible component $U$ of $V_{\varepsilon n, n^{1+\delta}}$. As shown by \cite{KLPS14}, the irreducible components are in one-to-one correspondence with subsets of $[n]\times [n]$ of size $n^{1+\delta}$ corresponding to possible supports of the sparse matrix $S$.

As we observe in \autoref{rmk: fixed sparsity pattern}, it is somewhat simpler to show that each of these irreducible components has an equation with a polynomial degree bound. However, since the number of such irreducible components is exponentially large, it is not clear if there is a single  equation for the whole variety which is of polynomially bounded degree. We do manage to reverse the order of quantifiers and prove such an upper bound in \autoref{thm:poly-deg-bound}. This suggests that the set of non-rigid matrices is much less complex than what one may suspect given the results of \cite{KLPS14, GHIL16}.

\subsection{Circuits for linear transformations}

The original motivation for defining rigidity was in the context of proving lower bounds for algebraic circuits \cite{Valiant77}. If $A \in \F^{n \times n}$ is an $(\varepsilon n, n^{1+\delta})$-rigid matrix, for any $\varepsilon, \delta >0$, then the linear transformation represented by $A$ cannot be computed by an algebraic circuit of depth $O(\log n)$ and size $O(n)$.

Every algebraic circuit computing a linear transformation is without loss of generality a \emph{linear} circuit. A linear circuit is a directed acyclic graph that has $n$ inputs labeled $X_1, \ldots, X_n$ and $n$ output nodes. Each edge is labeled by a scalar $\alpha \in \F$. Each node computes a linear function in $X_1, \ldots, X_n$ defined inductively. An internal node $u$ with children, $v_1, \ldots, v_k$, connected to it by edges labeled $\alpha_1, \ldots, \alpha_k$, computes the linear function $\sum_i \alpha_i \ell_{v_i}$, where $\ell_{v_i}$ is the linear function computed by $v_i$, $1 \le i\le k$. The size of the circuit is the number of edges in the circuit.

It is possible to use similar techniques to those used in the proof of \autoref{thm:poly-deg-bound} in order to prove a polynomial upper bound on an equation for a variety containing all matrices $A \in \F^{n \times n}$ whose corresponding linear transformation can be computed by an algebraic circuit of size at most $n^2/200$ (even without restriction on the depth). Note that this is nearly optimal as any such linear transformation can be computed by a circuit of size $n^2$. More formally, we show the following.  
\begin{theorem}
\label{thm:deg-bound-lin-ckt}
Let $\F$ be a field of size at least $n^2$. For every large enough $n$, there exists a non-zero polynomial $Q \in \F[x_{1,1}, \ldots, x_{n,n}]$, of degree at most $n^3$, which is a non-trivial equation for matrices which are computed by algebraic circuit of size at most $n^2/200$.
\end{theorem}

Our proofs are based on a dimension counting arguments, and are therefore non-constructive and do not give explicit equations for the relevant varieties. It thus remains a very interesting open problem to provide explicit low-degree equations for any of the varieties considered in this paper. Here ``explicit'' means a polynomial which has arithmetic circuits of size $\poly(n)$.\footnote{Although one may consider other, informal notions of explicitness which could nevertheless be helpful.} The question of whether such equations exists has a win-win flavor: if they do, this can aid in explicit constructions of rigid matrices, and on the other hand, if all equations are hard, we have identified a family of polynomials which requires super-polynomial arithmetic circuits. Assuming the existence of a polynomial time algorithm for polynomial identity testing, we are able to make this connection formal.

Let $\PIT$ denote the set of strings which describe arithmetic circuits (say, over $\C$) which compute the zero polynomial. It is well known that $\PIT \in \coRP$. Kabanets and Impagliazzo \cite{KI04} proved that certain circuit lower bounds follow from the assumption that $\PIT \in \P$. As a corollary to \autoref{thm:deg-bound-lin-ckt}, we are able to prove theorem of a similar kind.

\begin{corollary}\label{cor:win-win}
Suppose $\PIT \in \P$. Then at least one of the following is true:
\begin{enumerate}
\item There exists a family of $n$-variate polynomials of degree $\poly(n)$ over $\C$, which can be computed (as its list of coefficients, given the input $1^n$) in $\PSPACE$, which does not have polynomial size constant free arithmetic circuits.
\item there exists a family of matrices, constructible in polynomial time with an $\NP$ oracle (given the input $1^n$), which requires linear circuits of size $\Omega(n^2)$.
\end{enumerate}
\end{corollary}

A \emph{constant free arithmetic circuit} is an arithmetic circuit which is only allowed to use the constants $\{0,\pm 1\}$.

A different way to interpret \autoref{cor:win-win} is as saying that at least one of the following three lower bound results hold: either $\PIT \not \in \P$, or (at least) one of the two circuit lower bounds stated in the corollary. We emphasize that the result holds under \emph{any} (even so-called \emph{white box}) derandomization of $\PIT$.

Our statement is similar to, but incomparable with the result of Kabanets and Impagliazzo \cite{KI04} who proved that if $\PIT \in \P$ then either the permanent does not have polynomial size constant free arithmetic circuits, or $\NEXP \not\subseteq \P/\poly$.

Since $(\varepsilon n, \varepsilon n^2)$-rigid matrices have linear circuit of size $3\varepsilon n^2$, the last item of \autoref{cor:win-win} in particular implies a conditional construction of $(\Omega(n), \Omega(n^2))$-rigid matrices (it is also possible to directly use \autoref{thm:poly-deg-bound} instead of \autoref{thm:deg-bound-lin-ckt} to deduce this result). Unconditional constructions of rigid matrices in polynomial time with an $\NP$ oracle were recently given in \cite{AC19, BHPT20}. However, the rigidity parameters in these papers are not enough to imply circuit lower bounds (furthermore, even optimal rigidity parameters are not enough to imply $\Omega(n^2)$ lower bounds for general linear circuits).

Since it is widely believed that $\PIT \in \P$, the answer to which of the last two items of \autoref{cor:win-win} holds boils down to the question of whether there exists an equation for non-rigid matrices of degree $\poly(n)$ and circuit size $\poly(n)$. If determining if a matrix is rigid is $\coNP$-hard (as is known for some restricted ranges of parameters \cite{MS10}), it is tempting to also believe that the equations should not be easily computable, as they provide ``proof'' for rigidity which can be verified in randomized polynomial time. However, it could still be the case that those equations that have polynomial size circuits only prove the rigidity of ``easy'' instances.

\subsection{Some basic notions in algebraic geometry}
\label{sec:ag}

For completeness, in this section we define some basic notions in algebraic geometry. A reader who is familiar with this topic may skip to the next section.

Let $\F$ be an algebraically closed field. A set $V \subseteq \F^n$ is called an \emph{affine variety} if there exist polynomials $f_1, \ldots, f_t \in \F[x_1, \ldots, x_n]$ such that $V = \{\vx : f_1(\vx) = f_2(\vx) = \cdots = f_t(\vx) = 0\}$. For convenience, in this paper we often refer to affine varieties simply as varieties.

For each variety $V$ there is a corresponding ideal $\Ideal(V) \subseteq \F[x_1, \ldots, x_n]$ which is defined as
\[
\Ideal(V) := \{ f \in \F[x_1, \ldots, x_n] : f(\vx) = 0 \text{ for all } \vx \in V \}.
\]
Conversely, for an ideal $I \subseteq \F[x_1, \ldots, x_n]$ we may define the variety
\[
\Variety(I) = \{\vx : f(\vx) = 0  \text{ for all } f \in I \}.
\]

Given a set $A \subseteq \F^n$ we may similarly define the ideal $\Ideal(A)$. The (Zariski) \emph{closure} of a set $A$, denoted $\overline{A}$, is the set $\Variety(\Ideal(A))$. In words, the closure of $A$ is the set of common zeros of all the polynomials that vanish on $A$. It is also the smallest variety with respect to inclusion which contains $A$. By construction, $\overline{A}$ is a variety, and a polynomial which vanishes everywhere on $A$ is also vanishes on $\overline{A}$.

Over $\C$, it is instructive to think of the Zariski closure of $A$ as the usual Euclidean closure. In fact, for the various sets $A$ we consider in this paper (which correspond to sets of ``low complexity'' objects, e.g., non-rigid matrices or matrices which can be computed with a small circuit), it can be shown that these two notions of closure coincide (see, e.g., Section 4.2 of \cite{BI17}).

A variety $V$ is called \emph{irreducible} if it cannot be written as a union $V = V_1 \cup V_2$ of varieties $V_1, V_2$ that are properly contained in $V$. Every variety can be uniquely written as a union $V = V_1 \cup V_2 \cup \cdots \cup V_m$ of irreducible varieties. The varieties $V_1, \ldots, V_m$ are then called the \emph{irreducible components} of $V$.

\section{Degree Upper Bound for Non-Rigid Matrices}
\label{sec:rigid}

In this section, we prove \autoref{thm:poly-deg-bound}. A key component of the proof is the use of the following construction, due to Shpilka and Volkovich, which provides an explicit low-degree polynomial map on a small number of variables, which contains all sparse matrices in its image. For completeness, we provide the construction and prove its basic property.

\begin{lemma}[\cite{SV15}]
\label{lem:sv-generator}
Let $\F$ be a field such that $|\F|>n$. Then for all $k\in \N$, there exists an explicit polynomial map $\SV_{n,k} (\vx,\vy) : \F^{2k} \to \F^n$ of degree at most $n$ such that for any subset $T = \{i_1, \ldots, i_k\} \subseteq [n]$ of size $k$, there exists a setting $\vy=\valpha$ such that $\SV(\vx,\valpha)$ is identically zero on every coordinate $j \not\in T$, and equals $x_j$ in coordinate $i_j$ for all $j \in [k]$.
\end{lemma}

\begin{proof}
Arbitrarily pick distinct $\alpha_1, \ldots \alpha_n \in \F$, and let $u_1, \ldots, u_n$ be their corresponding Lagrange's interpolation polynomials, i.e., polynomials of degree at most $n-1$ such that $u_i(\alpha_j) = 1$ if $j=i$ and $0$ otherwise (more explicitly, $u_i(z) = \frac{\prod_{j \neq i} (z-\alpha_j)}{\prod_{j \neq i}(\alpha_i - \alpha_j)}$).

Let $P_i (x_1, \ldots, x_k, y_1, \ldots, y_k) = \sum_{j=1}^k u_i (y_j) \cdot x_j$, and finally let
\[
\SV_{n,k} (\vx, \vy) = (P_1(\vx, \vy), \ldots, P_n(\vx, \vy)).
\]
It readily follows that given $T=\{i_1, \ldots, i_k\}$ as in the statement of the lemma, we can set $y_j = \alpha_{i_j}$ for $j \in [k]$ to derive the desired conclusion. The upper bound on the degree follows by inspection.
\end{proof}

As a step toward the proof of \autoref{thm:poly-deg-bound}, we show there is a polynomial map on much fewer than $n^2$ variables with degree polynomially bounded in $n$ such that its image contains every non-rigid matrix. In the next step, we show that the image of every such polynomial map has an equation of degree $\poly(n)$. 
\begin{lemma}
\label{lem:map-for-rigid}
There exists an explicit polynomial map $P : \F^{4 \varepsilon n^2} \to \F^{n \times n}$, of degree at most $n^2$, such that every matrix $M$ which is not $(\varepsilon n, \varepsilon n^2)$ rigid lies in its image.
\end{lemma}

\begin{proof}
Let $k = \varepsilon n^2$ and let $\vu, \vv, \vx, \vy$ denote disjoint tuples of $k$ variables each. 

Let $U$ be a symbolic $n \times \varepsilon n$ matrix whose entries are labeled by the variables $\vu$, and similarly let $V$ be a symbolic $\varepsilon n \times n$ matrix labeled by $\vv$. Let $\mathrm{UV} (\vu, \vv) : \F^{2k} \to \F^{n \times n}$ be the degree 2 polynomial map defined by the matrix multiplication $UV$.

Finally, let $P: \F^{4k} \to \F^{n \times n}$ be defined as
\[
P(\vu,\vv,\vx,\vy) = \mathrm{UV}(\vu,\vv) + \SV_{n^2,k} (\vx,\vy),
\]
where $\SV_{n^2,k}$ is as defined in Lemma \ref{lem:sv-generator}.

Suppose now $M$ is a non-rigid matrix, i.e., $M = R+S$ for $R$ of rank $\varepsilon n$ and $S$ which is $\varepsilon n^2$-sparse. Decompose $R = U_0 V_0$ for $n \times \varepsilon n$ matrix $U_0$ and $\varepsilon n \times n$ matrix $V_0$. Let $T$ denote the support of $S$. For convenience we may assume $|T|=k$ (otherwise, pad with zeros arbitrarily). Let $\valpha \in \F^k$ denote the setting for $\vy$ in $\mathrm{SV}_{n^2,k}$ which maps $x_1, \ldots, x_k$ to $T$, and let $\vs=(s_1, \ldots, s_k)$ denote the non-zero entries of $S$. Then
\[
P(U_0, V_0, \vs, \valpha) = U_0 V_0 + S = R+S = M. \qedhere
\]
\end{proof}

To complete the proof of \autoref{thm:poly-deg-bound}, we now argue that the image of any polynomial map with parameters as in \autoref{lem:map-for-rigid} has an equation of degree at most $n^3$. 

\begin{proof}[Proof of Theorem~\ref{thm:poly-deg-bound}]
Let $V_1$ denote the subspace of polynomials over $\F$ in $n^2$ variables of degree at most $n^3$. Let $V_2$ denote the subspace of polynomials over $\F$ in $4\varepsilon n^2$ variables of degree at most $n^5$. Let $P$ be as in Lemma~\ref{lem:map-for-rigid}, and consider the linear transformation $T : V_1 \to V_2$ given by $Q \mapsto Q \circ P$, where $Q \circ P$ denotes the composition of the polynomial $Q$ with the map $P$, i.e., $(Q \circ P) (\vx) = Q(P(\vx))$ (indeed, observe that since $\deg (Q) \le n^3$ and $\deg(P) \le n^2$, it follows that $\deg Q \circ P \le n^5$).

We have that $\dim(V_1) =  \binom{n^3 + n^2}{n^2} \ge n^{n^2}$, whereas $\dim(V_2) = \binom{4\varepsilon n^2 + n^5}{4\varepsilon n^2} \le (2n^5)^{4 \varepsilon n^2} < \dim(V_1)$  by the choice of $\varepsilon$, so that there exists a non-zero polynomial in the kernel of $T$, that is, $ 0 \neq Q_0 \in V_1$ such that $Q_0 \circ P \equiv 0$.

It remains to be shown that for any non-rigid matrix $M$, $Q_0(M) = 0$. Indeed, let $M$ be a non-rigid matrix. By Lemma \ref{lem:map-for-rigid}, there exist $\vbeta \in \F^{4\varepsilon n^2}$ such that $P(\vbeta) = M$. Thus, $Q_0(M) = Q_0(P(\vbeta)) = Q_0\circ P (\beta) =  0$, as $Q_0 \circ P \equiv 0$.
\end{proof}

\begin{remark}\label{rmk: fixed sparsity pattern}
If the support of the sparse matrix is fixed a-priori to some set $S \subseteq [n] \times [n]$ of cardinality at most $\epsilon n^2$, then it is easier to come up with a universal map $\tilde{P}$ from $\F^{3\epsilon n^2} \mapsto \F^{n \times n}$ such that every matrix $M$ whose rank can be reduced to at most $\epsilon n$ by changing entries in the set $S$ is contained in the image of $\tilde{P}$. Just consider $\tilde{P}(\vw, \vx, \vy) = \mathrm{UV}(\vu,\vv) + W$, where $W$ is a matrix such that for all $(i, j) \in [n]\times [n]$, if $(i, j) \in S$, then $W(i, j) = w_{i, j}$ and $W(i, j)$ is zero otherwise. Here, each $w_{i, j}$ is a distinct formal variable. Combined with the dimension comparison argument we used in the proof of \autoref{thm:poly-deg-bound}, it can be seen that there is a non-zero low degree polynomial $\tilde{Q}$ such that $\tilde{Q}\circ \tilde{P} \equiv 0$. This argument provides a (different) equation of polynomial degree for each irreducible component of the variety of non-rigid matrices. 
\end{remark}

\begin{remark}\label{rmk: semi explicit matrices}
It is possible to use the equation given in \autoref{thm:poly-deg-bound}, and using the methods of \cite{KLPS14}, to construct ``semi-explicit'' $(\varepsilon n, \varepsilon n^2)$-rigid matrices. These are matrices whose entries are algebraic numbers (over $\Q$) with short description, which are non-explicit from the computational complexity point of view. However, such constructions are also known using different methods (see Section 2.4 of \cite{Lokam09}).
\end{remark}

\section{Degree Upper bound for Matrices with a Small Circuit}\label{sec:circuit}
In this section, we prove \autoref{thm:deg-bound-lin-ckt}.  Our strategy, as before, is to observe that all matrices with a small circuit lie in the image of a polynomial map $P$ on a small number of variables and small degree. Circuits of size $s$ can have many different topologies and thus we first construct a ``universal'' linear circuit, of size $s' \le s^4$, that contains as subcircuits all linear circuits of size $s$. Importantly, $s'$ will affect the degree of $P$ but not its number of variables. We note that this construction of universal circuits is slightly different from similar constructions in earlier work, e.g., in \cite{R10b}; the key difference being that a naive use of ideas in \cite{R10b} to obtain the map $P$ seems to incur an asymptotic increase in the number of variables of $P$, which is unacceptable in our current setting.

\subsection{A construction of universal map for small linear circuits}

We now define a map $U(\vx,\vy)$ which is ``universal'' for size $s$ linear circuits, i.e., it contains in its image all $n \times n$ matrices $A$ whose corresponding linear transformation can be computed by a linear circuit of size at most $s$.

Let $s \ge n$. We first define a universal graph $G$ for size $s$. $G$ has a set $V_0$ of $n$ input nodes labeled $X_1, \ldots X_n$ and a set $V_{s+1}$ of $n$ designated output nodes. In addition, $G$ is composed of $s$ disjoint sets of vertices $V_1, \ldots, V_s$, each contains $s$ vertices.

Each vertex $v \in V_i$, for $0 \le i \le s+1$, has as its children all vertices $u \in V_j$ for all $0 \le j < i$. It is clear than every directed acyclic graph with $s$ edges (and hence at most $s$ vertices, and depth at most $s$) can be (perhaps non-uniquely) embedded in $G$ as a subgraph.

We now describe the edge labeling. Let $s' \le s^4$ be the number of edges in $V$ and let $e_i$ denote the $i$-th edge, $1 \le i \le s'$. The edge $e_i$ is labeled by the $i$-th coordinate of the map $\SV_{s', s} (\vx, \vy)$ given in \autoref{lem:sv-generator}.

Thus, the graph $G$ with this labeling computes a linear transformation (over the field $\F(\vx,\vy)$) in the variables $X_1, \ldots, X_n$. More explicitly, the $(i,j)$-th entry of the matrix $U(\vx,\vy)$ representing this linear transformation is given by the sum, over all paths from $X_i$ to the $j$-th output node, of the product of the edge labels on that path. This entry is a polynomial in $\vx, \vy$, so that we can think of $U$ as a polynomial map from $\F^{2s}$ to $\F^{n^2}$.

\begin{lemma}\label{lem:universal-map-linear}
The map $U(\vx,\vy)$ defined above contains in its image all $n \times n$ matrices $A$ whose corresponding linear transformation can be computed by a linear circuit of size at most $s$. The degree of $U$ is at most $s' \cdot (s+1)$.
\end{lemma}

\begin{proof}
Let $A$ be a matrix whose linear transformation is computed by a size $s$ circuit $C$. The graph of $C$ can be embedded as a subgraph in the graph $G$ constructed above (if the embedding is not unique, pick one arbitrarily). Let $e_{i_1}, \ldots, e_{i_s}$ be the edges of this subgraph, and let $\vbeta=(\beta_1, \ldots, \beta_s)$ be their corresponding labels in $C$. By the properties of the map $\SV_{s',s}(\vx,\vy)$ given in \autoref{lem:sv-generator}, it is possible to set the tuple of variables $\vy$ to field elements $\alpha_1, \dots, \alpha_s$ such that the $j$-th coordinate of $\SV(\vbeta, \valpha)$ equals $\beta_i$ if $j=i_k$ for some $1 \le k \le s$ the $0$ otherwise. Observe that under this labeling of the edges, the circuit $G$ computes the same transformation as the circuit $C$. Hence $U(\vbeta, \valpha) = A$.

To upper bound the degree of $U$, note that each edge label in $G$ is a polynomial of degree $s'$, and each path is of length at most $s+1$.
\end{proof}

\subsection{Low degree equations for small linear circuits}

Analogous to the proof of \autoref{thm:poly-deg-bound}, we now observe via a dimension counting argument that the image of the polynomial map $U(\vx, \vy)$ has a  equation of degree at most $n^3$. This would complete the proof of \autoref{thm:deg-bound-lin-ckt}.

\begin{proof}[Proof of \autoref{thm:deg-bound-lin-ckt}]
As before, let $V_1$ denote the subspace of polynomials over $\F$ in $n^2$ variables of degree at most $n^3$. Let $V_2$ denote the subspace of polynomials over $\F$ in $n^2/100$ variables of degree at most $n^{30}$. 
Let $U$ be the map given by \autoref{lem:universal-map-linear} for $s=n^2/200$ so that $s' \le n^8$, and the degree of $U$ is at most $s'(s+1) \le n^{10}$. Now, consider the linear transformation $T : V_1 \to V_2$ given by $Q \mapsto Q \circ U$. 

Once again, we compute that $\dim(V_1) =  \binom{n^3 + n^2}{n^2} \ge n^{n^2}$, whereas $\dim(V_2) = \binom{n^2/100 + n^{30}}{n^2/100} \le (2n^{30})^{n^2/100} < \dim(V_1)$, so that there exists a non-zero polynomial in the kernel of $T$, that is, $ 0 \neq Q_0 \in V_1$ such that $Q_0 \circ U \equiv 0$.

By \autoref{lem:universal-map-linear}, if $A$ has a circuit of size $n^2/200$, it is in the image of $U$, so that $Q_0 (A) = 0$.
\end{proof}

\section{Degree Upper Bound for Three Dimensional Tensors}

Another algebraic object which is closely related to proving circuit lower bounds is the set of three dimensional tensors of high rank. A three dimensional tensor of rank at least $r$ implies a lower bound of $r$ on an arithmetic circuit computing the bi-linear function associated with the tensor. Our arguments also provide polynomial degree upper bounds for the set of tensors of (border) rank at most $n^2/300$.

\begin{lemma}\label{lem:univ 3 tensor}
Let $\F$ be any field. There is a polynomial map $P:\F^{n^2/100} \to \F^{n^3}$ of degree at most $3$ such that for every $3$ dimensional tensor $\tau:[n]^3 \to \F$ of rank at most $n^{2}/300$ lies in its image. 
\end{lemma}
\begin{proof}
This follows immediately from the definition.

Indeed,
let $r = n^{2}/300 $. Let $\vu_1, \ldots, \vu_r, \vv_1, \ldots, \vv_r, \vw_1, \ldots, \vw_r$  be disjoint $n$ tuples of variables. Let $U$ be a tensor of rank at most $r$ over the ring  $\F[\vu_1, \ldots, \vu_r, \vv_1, \ldots, \vv_r, \vw_1, \ldots, \vw_r]$ defined as follows. 
\[
U(\vu, \vv, \vw) = \sum_{i = 1}^r \vu_i\otimes \vv_i \otimes \vw_i \, .
\]
From the definition of $U$, it can be readily observed that for every tensor $\tau : \F^{[n]^3}\to \F$ of rank at most $r$, there is a setting $\valpha, \vbeta, \vgamma$ of the variables in $\vu, \vv, \vw$ respectively such that $U(\valpha, \vbeta, \vgamma) = \tau$. Moreover,  each of the coordinates of $U$ is a polynomial of degree equal to three in the variables in $\vu, \vv, \vw$. Let $P$ be the degree three polynomial map which maps the variables $\vu_1, \ldots, \vu_r, \vv_1, \ldots, \vv_r$ and $\vw_1, \ldots, \vw_r$ to the coordinates of $U$. 
\end{proof}
We now argue that for every polynomial map $P$ given by \autoref{lem:univ 3 tensor} has an equation of not too large degree. 
\begin{theorem}\label{thm:deg bound 3 tensor}
Let $\F$ be any field. There exists a non-zero polynomial $Q \in \F[x_{1,1, 1}, \ldots, x_{n,n, n}]$, of degree at most $n^4$, which is a non-trivial equation for three dimensional tensors $\tau:[n]\times [n] \times [n] \mapsto \F$ of rank at most $n^2/300$. 
\end{theorem}
\begin{proof}
As before, let $V_1$ denote the subspace of polynomials over $\F$ in $n^3$ variables of degree at most $n^4$ and let $V_2$ denote the subspace of polynomials over $\F$ in $n^3/100$ variables of degree at most $3n^4$. Let $P$ be the map given by \autoref{lem:univ 3 tensor}.  Now, consider the linear transformation $T : V_1 \to V_2$ given by $Q \mapsto Q \circ P$. 

Observe that $\dim(V_1) =  \binom{n^4 + n^3}{n^3} \ge n^{n^3}$, whereas $\dim(V_2) = \binom{n^3/100 + 3n^4}{n^3/100} \le (2n^{4})^{n^3/100} < \dim(V_1)$, so that there exists a non-zero polynomial in the kernel of $T$, that is, $ 0 \neq Q_0 \in V_1$ such that $Q_0 \circ P \equiv 0$.

By \autoref{lem:univ 3 tensor}, if $\tau$ is a tensor of rank at most $n^2/300$, then it is in the image of $P$, and thus $Q_0(\tau) = 0$.
\end{proof}

The arguments here also generalize to tensors in higher dimensions. In particular, the following analog of \autoref{lem:univ 3 tensor} is true.

\begin{lemma}\label{lem:univ d tensor}
Let $\F$ be any field. Then, for all $n, d \in \N$, there is a polynomial map $P:\F^{n^{d-1}/100} \to \F^{n^d}$ of degree at most $d$ such that for every $d$ dimensional tensor $\tau:[n]^{\otimes d} \to \F$ of rank at most $n^{d-1}/100d$ lies in its image. 
\end{lemma}

Combining this lemma with a dimension comparison argument analogous to that in the proof of \autoref{thm:deg bound 3 tensor} gives the following theorem. We skip the details of the proof. 
\begin{theorem}\label{thm:deg bound d tensor}
For every field $\F$ and for all $n, d \in \N$,  there exists a non-zero polynomial $Q$ on $n^d$ variables and degree at most $n^{2d}$,  which is a non-trivial equation for  $d$ dimensional tensors $\tau:[n]^{\otimes d} \to \F$ of rank at most $n^{d-1}/100 d$. 
\end{theorem}

We remark that a similar methods can be used to prove the existence of an equation of degree $\poly(n)$ for three dimensional tensors of \emph{slice rank} (see, e.g., \cite{BIPLS19}) at most, say, $n/1000$. The existence of such an equations was proved (using different techniques) in \cite{BIPLS19}.

\section{Applications to Circuit Lower Bounds}

In this section we prove \autoref{cor:win-win}. The strategy of the proof is simple: the proof of \autoref{thm:deg-bound-lin-ckt} implies a $\PSPACE$ algorithm which produces a sequence of polynomials which are equations for the set of matrices with small linear circuits. If those equations require large circuits, we are done, and if not, then there exists an equation with small circuits which (assuming $\PIT \in \P$) can be found using an $\NP$-oracle. Using, once again, the assumption that $\PIT \in \P$, we can also find deterministically a matrix on which the equation evaluates to non-zero, which implies the matrix requires large linear circuits.

There are some technical difficulties involved with this plan which we now describe. The first problem is that even arithmetic circuits of small size can have large description as bit strings, due to the field constants appearing in the circuits. To prevent this issue, we only consider \emph{constant free} arithmetic circuits, which are only allowed inputs labeled by $\{0,\pm 1\}$ (but can still compute other constants in the circuit using arithmetic operations).

The second problem is that, in order to be able to find a non-zero of the equation in the last step of the algorithm (using the mere assumption that $\PIT \in \P$), we need not only the size of the circuit but also its \emph{degree} to be bounded by $\poly(n)$. Of course, by \autoref{thm:deg-bound-lin-ckt} the exists such a circuit, but we need to be able to prevent a malicious prover from providing us with a $\poly(n)$ size circuit of exponential degree, and it is not known how to compute the degree of a circuit in deterministic polynomial time, even assuming $\PIT \in \P$. To solve this issue, we use an idea of Malod and Portier \cite{MP08}, who showed that any polynomial with circuit of size $\poly(n)$ and degree $d$ also has a \emph{multiplicatively disjoint} (MD) circuit of size $\poly(n,d)$. An MD circuit is a circuit in which any multiplication gates multiplies two disjoint subcircuits. This is a syntactic notion which is easy to verify efficiently and deterministically, and an MD circuit of size $s$ is guaranteed to compute a polynomial of degree at most $s$.

A final technical issue is that the notion of MD circuits does not fit perfectly within the framework of constant free circuits. Therefore we use the notion of ``almost MD'' circuits, which allow for the case which the inputs to a multiplciation gates are not disjoint, as long as at least one of them is the root of a subcircuit in which only constants appear.

\begin{definition}\label{def:almost-MD}
We say a gate $v$ in a circuit is \emph{constant producing} (CP) if in the subcircuit rooted at $v$, all input nodes are field constants. 

An \emph{almost-MD circuit} is a circuit where every multiplication gate either multiplies two disjoint subcircuits, or at least one of its children is constant producing.
\end{definition}

\begin{lemma}\label{lem:almost-MD-ckts}
Suppose $f$ is an $n$-variate polynomial of degree $\poly(n)$ which has a constant free arithmetic circuit of degree $\poly(n)$. Then $f$ has a constant free almost-MD circuit of size $\poly(n)$.
\end{lemma}

\begin{proof}
Let $C_0$ be a constant free arithmetic circuit for $f$. We first homogenize the circuit $C_0$ to obtain a circuit $C_1$ (a homogeneous circuit is a circuit in which every gate computes a homogeneous polynomial \cite{SY10}). Since $C_1$ is homogeneous, all the gates which compute non-zero field constants are CP gates. We then eliminate all gates which compute constants by allowing the edges entering sum gates to be labeled by field scalars, and interpreting a sum gate as computing a linear combination whose coefficients are given by the edge labels. We call this circuit $C_2$. This step does not maintain constant-freeness. However, every label appearing on the edges of $C_2$ was computed in $C_1$, so it can be computed by a constant-free arithmetic circuit of polynomial size.

We now do the transformation detailed in \cite{MP08} to $C_2$ to obtain an MD circuit $C_3$, which has labels on the edges. This step does not produce new constants. Finally, we convert $C_3$ to an almost-MD constant free circuit $C_4$, by re-computing every label appearing on the edge using a fresh subcircuit for each label, and rewiring the circuit (which will convert the circuit from an MD circuit to an almost MD circuit). These subcircuits are guaranteed to have polynomial size constant free circuits since these constant were all computed in $C_0$, which keeps the total size $\poly(n)$.
\end{proof}

For circuits which compute low-degree polynomials, the mere existence of an algorithm for the decision version of PIT allows one to construct an algorithm for the search version.

\begin{lemma}\label{lem:pit-witness}
Suppose $\PIT \in \P$. Then there is a polynomial time algorithm that given a non-zero almost-MD arithmetic circuit $C$ of size $s$ computing an $n$-variate polynomial, finds in time $\poly(n,s)$ an element $\va \in \C^n$ such that $C(\va) \neq 0$.
\end{lemma}

\begin{proof}
We abuse notation by denoting by $C$ also the polynomial computed by the circuit $C$. Note that since $C$ is almost-MD, the degree of $C$ is at most $s$. Thus, there exists $a_1 \in \{0,1,\ldots, s\}$ such that $C(a_1, x_2, \ldots, x_n)$ is a non-zero polynomial in $x_2, \ldots, x_n$. By iterating over those $s+1$ values from $0$ to $s$ and using the assumption that $\PIT \in \P$, we can find such a value for $a_1$ in time $\poly(n,s)$. We then continue in the same manner with the rest of the variables.
\end{proof}

As we noted above, the assumption that $C$ is almost-MD was used in \autoref{lem:pit-witness} to bound the degree of the circuit. It is also useful because it is easy to decide in deterministic polynomial time whether a circuit is almost-MD. We now complete the proof of \autoref{cor:win-win}.

\begin{proof}[Proof of \autoref{cor:win-win}]
For every $n$, the proof of \autoref{thm:deg-bound-lin-ckt} provides an equation $Q_n$ for the set of $n \times n$ matrices with small linear circuits. This polynomial can be found by solving a linear system of equations in a linear space whose dimension is $\exp(\poly(n))$. Using standard, small space algorithm for linear algebra \cite{BvzGH82, ABO99}, this implies that there exists a fixed $\PSPACE$ algorithm which, on input $1^n$, outputs the list of coefficients of the polynomial $Q_n$.

Consider now the family $\{Q_n\}_{n \in \N}$. If for any constant $k \in \N$ there exist infinitely many $n \in \N$ such that $Q_n$ requires circuits of size at least $n^k$, it follows (by definition) that the $\PSPACE$ algorithm above outputs a family of polynomials with super-polynomial constant-free arithmetic circuits.

We are thus left to consider the case that there exists a constant $k \in \N$ such that for all large enough $n \in \N$, $Q_n$ can be computed by circuits of size $n^k$. By \autoref{lem:almost-MD-ckts}, we may assume without loss of generality that these circuits are almost-MD circuits. Further suppose $\PIT \in \P$. We will show how to construct a matrix in polynomial time with an $\NP$ oracle which requires large linear circuits.

Consider the language $L$ of pairs $(1^n, x)$ such that there exists a string $y$ of length at most $n^k$ such that $xy$ describes an almost-MD circuit $C$ such that $C$ is non-zero, and $C \circ U \equiv 0$, where $U$ is the polynomial map given in the proof of \autoref{thm:deg-bound-lin-ckt}.
 
Assuming $\PIT \in \P$, the language $L$ is in $\NP$, and by assumption for every large enough $n$ there exists such a circuit. Thus, we can use the $\NP$ oracle to construct such a circuit $C$ bit by bit. Finally, using \autoref{lem:pit-witness} we can output a matrix $M$ such that $C(M) \neq 0$.

By the properties of the circuit $C$ and the map $U$, $M$ does not have linear circuits of size less than $n^2/200$. 
\end{proof}

Many variations of \autoref{cor:win-win} can be proved as well, with virtually the same proof. By slightly modifying the language $L$ used in the proof, it is possible to prove the same result even under the assumption $\PIT \in \NP$ (recall that $\PIT \in \coRP$). A similar statements also holds over finite fields of size $\poly(n)$, in which case the proof is simpler since there are no issues related to the bit complexity of the first constants. Finally, an analog of \autoref{cor:win-win} also holds for tensor rank, by using \autoref{thm:deg bound 3 tensor} instead of \autoref{thm:deg-bound-lin-ckt}: that is, assuming $\PIT \in \P$, either there exists a construction of a hard polynomial in $\PSPACE$, or an efficient construction with an $\NP$ oracle of a 3-dimensional tensor of rank $\Omega(n^2)$. We remark that for tensors of large rank there are no analogs of \cite{AC19, BHPT20}, i.e., there do not exist even constructions with an $\NP$ oracle of tensors with slightly super-linear rank.

\bibliographystyle{customurlbst/alphaurlpp}
\bibliography{references}

\newcommand{\etalchar}[1]{$^{#1}$}
\begin{thebibliography}{BvzGH82}

\bibitem[ABO99]{ABO99}
Eric Allender, Robert Beals, and Mitsunori Ogihara.
\newblock \href {http://dx.doi.org/10.1007/s000370050023} {The Complexity of
  Matrix Rank and Feasible Systems of Linear Equations}.
\newblock {\em Comput. Complex.}, 8(2):99--126, 1999.

\bibitem[AC19]{AC19}
Josh Alman and Lijie Chen.
\newblock \href {http://dx.doi.org/10.1109/FOCS.2019.00067} {Efficient
  Construction of Rigid Matrices Using an {NP} Oracle}.
\newblock In {\em \FOCS{2019}}, pages 1034--1055. {IEEE} Computer Society,
  2019.

\bibitem[BHPT20]{BHPT20}
Amey Bhangale, Prahladh Harsha, Orr Paradise, and Avishay Tal.
\newblock \href {https://eccc.weizmann.ac.il/report/2020/075} {Rigid Matrices
  From Rectangular PCPs}.
\newblock {\em Electronic Colloquium on Computational Complexity {(ECCC)}},
  27:75, 2020.

\bibitem[BI17]{BI17}
Markus Bl\"{a}ser and Christian Ikenmeyer.
\newblock \href
  {http://pcwww.liv.ac.uk/~iken/teaching_sb/summer17/introtogct/gct.pdf}
  {Introduction to geometric complexity theory}.
\newblock Lecture notes, 2017.

\bibitem[BIL{\etalchar{+}}19]{BIPLS19}
Markus Bl{\"{a}}ser, Christian Ikenmeyer, Vladimir Lysikov, Anurag Pandey, and
  Frank{-}Olaf Schreyer.
\newblock \href {http://arxiv.org/abs/1911.02534} {Variety Membership Testing,
  Algebraic Natural Proofs, and Geometric Complexity Theory}.
\newblock {\em CoRR}, abs/1911.02534, 2019.
\newblock Pre-print available at \href {http://arxiv.org/abs/1911.02534}
  {\path{arXiv:1911.02534}}.

\bibitem[BvzGH82]{BvzGH82}
Allan Borodin, Joachim von~zur Gathen, and John~E. Hopcroft.
\newblock \href {http://dx.doi.org/10.1016/S0019-9958(82)90766-5} {Fast
  Parallel Matrix and {GCD} Computations}.
\newblock {\em Inf. Control.}, 52(3):241--256, 1982.

\bibitem[FSV18]{FSV18}
Michael~A. Forbes, Amir Shpilka, and Ben~Lee Volk.
\newblock \href {http://dx.doi.org/10.4086/toc.2018.v014a018} {Succinct Hitting
  Sets and Barriers to Proving Lower Bounds for Algebraic Circuits}.
\newblock {\em Theory of Computing}, 14(1):1--45, 2018.

\bibitem[GHIL16]{GHIL16}
Fulvio Gesmundo, Jonathan~D. Hauenstein, Christian Ikenmeyer, and J.~M.
  Landsberg.
\newblock \href {http://dx.doi.org/10.1007/s10208-015-9258-8} {Complexity of
  Linear Circuits and Geometry}.
\newblock {\em Foundations of Computational Mathematics}, 16(3):599--635, 2016.

\bibitem[GKSS17]{GKSS17}
Joshua~A. Grochow, Mrinal Kumar, Michael~E. Saks, and Shubhangi Saraf.
\newblock \href {http://arxiv.org/abs/1701.01717} {Towards an algebraic natural
  proofs barrier via polynomial identity testing}.
\newblock {\em CoRR}, abs/1701.01717, 2017.
\newblock Pre-print available at \href {http://arxiv.org/abs/1701.01717}
  {\path{arXiv:1701.01717}}.

\bibitem[KI04]{KI04}
Valentine Kabanets and Russell Impagliazzo.
\newblock \href {http://dx.doi.org/10.1007/s00037-004-0182-6} {{D}erandomizing
  Polynomial Identity Tests Means Proving Circuit Lower Bounds}.
\newblock {\em Computational Complexity}, 13(1-2):1--46, 2004.
\newblock \pSTOC{2003}.

\bibitem[KLPS14]{KLPS14}
Abhinav Kumar, Satyanarayana~V. Lokam, Vijay~M. Patankar, and Jayalal Sarma.
\newblock \href {http://dx.doi.org/10.1007/s00037-013-0061-0} {Using
  Elimination Theory to Construct Rigid Matrices}.
\newblock {\em Computational Complexity}, 23(4):531--563, 2014.

\bibitem[Lok09]{Lokam09}
Satyanarayana~V. Lokam.
\newblock \href {http://dx.doi.org/10.1561/0400000011} {Complexity Lower Bounds
  using Linear Algebra}.
\newblock {\em Foundations and Trends in Theoretical Computer Science},
  4(1-2):1--155, 2009.

\bibitem[MP08]{MP08}
Guillaume Malod and Natacha Portier.
\newblock \href {http://dx.doi.org/10.1016/j.jco.2006.09.006} {Characterizing
  Valiant's algebraic complexity classes}.
\newblock {\em J. Complex.}, 24(1):16--38, 2008.

\bibitem[MS01]{MulmuleyS01}
Ketan Mulmuley and Milind~A. Sohoni.
\newblock \href {http://dx.doi.org/10.1137/S009753970038715X} {Geometric
  Complexity Theory {I:} An Approach to the {P} vs. {NP} and Related Problems}.
\newblock {\em {SIAM} J. Comput.}, 31(2):496--526, 2001.

\bibitem[MS10]{MS10}
Meena Mahajan and Jayalal Sarma.
\newblock \href {http://dx.doi.org/10.1007/s00224-008-9136-8} {On the
  Complexity of Matrix Rank and Rigidity}.
\newblock {\em Theory Comput. Syst.}, 46(1):9--26, 2010.

\bibitem[Raz10]{R10b}
Ran Raz.
\newblock \href {http://dx.doi.org/10.4086/toc.2010.v006a007} {Elusive
  Functions and Lower Bounds for Arithmetic Circuits}.
\newblock {\em Theory of Computing}, 6(7):135--177, 2010.

\bibitem[SV15]{SV15}
Amir Shpilka and Ilya Volkovich.
\newblock \href {http://dx.doi.org/10.1007/s00037-015-0105-8} {Read-once
  polynomial identity testing}.
\newblock {\em Computational Complexity}, 24(3):477--532, 2015.
\newblock \pSTOC{2008}.

\bibitem[SY10]{SY10}
Amir Shpilka and Amir Yehudayoff.
\newblock \href {http://dx.doi.org/http://dx.doi.org/10.1561/0400000039}
  {Arithmetic Circuits: A survey of recent results and open questions}.
\newblock {\em Foundations and Trends in Theoretical Computer Science},
  5:207--388, March 2010.

\bibitem[Val77]{Valiant77}
Leslie~G. Valiant.
\newblock \href {http://dx.doi.org/10.1007/3-540-08353-7\_135} {Graph-Theoretic
  Arguments in Low-Level Complexity}.
\newblock In {\em \MFCS{1977}}, volume~53 of {\em Lecture Notes in Computer
  Science}, pages 162--176. Springer, 1977.

\end{thebibliography}

\end{document}